\documentclass[runningheads,USenglish]{llncs}

\usepackage{amssymb}
\usepackage{amsmath}
\usepackage{graphicx}
\usepackage[font=small,labelfont=bf]{caption}
\usepackage[labelfont=default]{subcaption}
\usepackage[disable]{todonotes}
\usepackage{hyperref}
\usepackage{enumitem}
\usepackage{comment}
\usepackage{cite}
\usepackage{xspace}

\title{On Arrangements of Orthogonal Circles\thanks{\lncsarxiv{The
      full version of this article is available at
      ArXiv~\cite{full-version-arxiv}. }{}M.K.\ was supported by
    DAAD; S.C.\ was supported by DFG grant WO$\,$758/11-1.}}

\author{Steven~Chaplick\inst{1}\lncsarxiv{\orcidID{0000-0002-1129-3289}}{} \and
  Henry~F\"orster\inst{2} \and Myroslav~Kryven\inst{1} \and
  Alexander~Wolff\inst{1}\lncsarxiv{\orcidID{0000-0001-5872-718X}}{}}

\institute{Universit\"at W\"urzburg, W\"urzburg, Germany \and
  Universit\"at T\"ubingen, T\"ubingen, Germany}

\authorrunning{S.~Chaplick et al.}

\graphicspath{{figures/}}

\newcommand{\lncsarxiv}[2]{#2} %

\let\doendproof\endproof
\renewcommand\endproof{~\hfill$\qed$\doendproof}

\newtheorem{observation}{Observation}
\newcommand{\D}{\displaystyle}

\newcommand{\ra}[1]{r_{#1}} %
\newcommand{\ce}[1]{C_{#1}} %

\newcommand{\arr}{\ensuremath{\mathcal{A}}\xspace} %
\newcommand{\subtend}[2]{\angle({#1}, {#2})} %

\usepackage{thm-restate}

\begin{document}

\maketitle

\centerline{\textit{Dedicated to Honza Kratochv\'il on his 60th birthday.}}

\begin{abstract}
  In this paper, we study arrangements of \emph{orthogonal circles},
  that is, arrangements of circles where every pair of circles must
  either be disjoint or intersect at a right angle.  Using geometric
  arguments, we show that such arrangements have only a
  linear number of faces.  This implies that
  \emph{orthogonal circle intersection graphs} have only a linear
  number of edges.  When we restrict ourselves to orthogonal
  \emph{unit} circles, the resulting class of intersection graphs is a
  subclass of penny graphs (that is, contact graphs of unit circles).
  We show that, similarly to penny graphs, it is NP-hard to recognize
  orthogonal unit circle intersection graphs.
\end{abstract}

\section{Introduction}

For the purpose of this paper, an \emph{arrangement} is a (finite)
collection of curves such as lines or circles in the plane.
The study of arrangements has a long history; for example,
Gr\"un\-baum~\cite{gruenbaum} studied arrangements of lines in the
projective plane.  
Arrangements of circles and other closed
curves have also been studied extensively~\cite{agarwal,alon,scheucher,kang-ross-arrangements,pinchasi}.  
An arrangement is \emph{simple} if no point of the plane belongs to more
than two curves and every two curves intersect.  A \emph{face} of an
arrangement~\arr in the projective or Euclidean plane $P$ 
is a connected component of the subdivision
induced by the curves in~\arr, that is, a face is a
component~of~$P \setminus \bigcup \arr$.

For a given type of curves, people have investigated the maximum
number of faces that an arrangement of such curves can form.
In 1826, Steiner~\cite{steiner1826}
showed that a simple arrangement of straight lines can have at most
${\binom{n}{2}} + {\binom{n}{1}} + {\binom{n}{0}}$ faces while an
arrangement of circles can have at most
$2\left({\binom{n}{2}} + {\binom{n}{0}} \right)$ faces.

Alon et al.~\cite{alon} and Pinchasi~\cite{pinchasi} studied the
number of \emph{digonal} faces, that is, faces that are bounded by two
edges, for various kinds of arrangements of circles.  For example, any
arrangement of $n$ unit circles has $O(n^{4/3}\log{n})$ digonal faces~\cite{alon}
and at most $n+3$ digonal faces if every pair
of circles intersects~\cite{pinchasi}, whereas arrangements of circles
with arbitrary radii have at most $20n-2$ digonal faces if every pair
of circles intersects~\cite{alon}.

The same arrangements can, however, have quadratically many \emph{triangular}
faces, that is, faces that are bounded by three edges. 
A lower bound example with quadratically many triangular faces can be constructed from a simple arrangement $\arr$ of lines  by
projecting it on a sphere (disjoint from the plane containing \arr) and having each line become a great circle.
This is always possible
since the line arrangement is simple; for more details see~\cite[Section 5.1]{felsner-book}. 
In this process we obtain $2p_3$ triangular faces, where $p_3$ is the number
of triangular faces in the line arrangement.
The great circles on the sphere can then be transformed into a circle
arrangement in a different plane using the stereographic projection.
This gives rise to an arrangement of circles with $2p_3$ triangular
faces in this plane.  F\"uredi and Pal\'asti~\cite{fueredi} provided
simple line arrangements with $n^{2}/3 + O(n)$ triangular faces.  With
the argument above, this immediately yields a lower bound of
$2n^{2}/3 + O(n)$ on the number of triangular faces of arrangements of
circles.  Felsner and Scheucher~\cite{scheucher} showed that this
lower bound is tight by proving that an arrangement of
\emph{pseudocircles} (that is, closed curves that can intersect at
most twice and no point belongs to more than two curves) can have at
most $2n^{2}/3 + O(n)$ triangular faces.

One can also specialize circle arrangements by
fixing an angle (measured as the angle 
between the two tangents at either intersection point)
at which each pair of intersecting circles
intersect; 
this was recently discussed by Eppstein~\cite{Eppstein-blog2018}.
In this paper, we consider arrangements of circles with the
restriction that each pair of circles must intersect at a right angle.
An arrangement of circles in which each intersecting pair intersect at a right angle is called \emph{orthogonal}.
We make the following simple observation regarding orthogonal circles;
see Fig.~\ref{fig:digon}.
\begin{observation}
\label{obs:basic}
Let $\alpha$ and $\beta$ be two circles
with centers $\ce{\alpha}$, $\ce{\beta}$ and radii
$\ra{\alpha}$, $\ra{\beta}$, respectively. 
Then $\alpha$ and $\beta$ are orthogonal if and only if
$\ra{\alpha}^2 + \ra{\beta}^2 = |\ce{\alpha}\ce{\beta}|^2$.
\end{observation}

\begin{figure}[tb]
  \centering
  \includegraphics{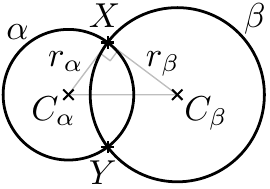}
  \caption{Circles~$\alpha$ and $\beta$ are orthogonal if and only if
    $\triangle\ce{\alpha}X\ce{\beta}$ is orthogonal.}
  \label{fig:digon}
\end{figure}
We discuss further basic properties of orthogonal circles in
Section~\ref{sec:preliminaries}.
In particular, in an arrangement of orthogonal circles no two circles
can touch and no three circles can intersect at the same point.

The main result of our paper is that arrangements of $n$ orthogonal
circles have at most $14n$ intersection points and at most $15n+2$ faces; see
Theorem~\ref{thm:orthogonalCircleArrangements} (in
Section~\ref{sec:complexity}).  This is different from arrangements
of orthogonal circular arcs, which can have quadratically many
quadrangular faces; see the arcs inside the blue square in
Fig.~\ref{fig:apollonian-just-parabolic}.
In Section~\ref{sec:counting-small-faces} we also consider small (that is, digonal and triangular) faces
and provide bounds on the number of such faces in arrangements
of orthogonal circles.

Given a set of geometric objects, their \emph{intersection graph} is a
graph whose vertices correspond to the objects and whose edges
correspond to the pairs of intersecting objects.  Restricting the
geometric objects to a certain shape restricts the class of graphs
that admit a representation with respect to this
shape.  For example,
graphs represented by disks in the Euclidean plane are called
\emph{disk intersection graphs}.
The special case of \emph{unit disk graphs}---intersection graphs of
unit disks---has been studied extensively.
Recognition of such graphs as well as many combinatorial problems
restricted to these graphs such as coloring, independent
set, and domination are all NP-hard~\cite{clarck1990};
see also the survey of Hlin\v{e}n\'{y} and
Kratochv\'{i}l~\cite{hk-rgdbs-DM01}.
Instead of restricting the radii of the disks, people have also
studied restrictions of the type of intersection.
If the disks are only allowed to touch, the corresponding graphs
are called \emph{coin graphs}.  Koebe's classical result says that the
coin graphs are exactly the planar graphs.
If all coins have the same size, the represented graphs are 
called \emph{penny graphs}.
These graphs have been studied extensively, too
\cite{Dumitrescu2011,CerioliFFP11,Eppstein2017TriangleFreePG}. 
For example, they are NP-hard to
recognize~\cite{kirkpatrick,gd-book}.

As with the arrangements above, we again consider a 
restriction on the intersection angle.
We define the \emph{orthogonal circle intersection graphs}
as the intersection graphs of arrangements of orthogonal circles.
In Section~\ref{sec:ocig}, we investigate properties of these graphs.
For example, similar to the proof of our linear bound on the number of intersection points for arrangements of orthogonal circles (Theorem~\ref{thm:orthogonalCircleArrangements}), 
we observe that such graphs have only a linear number of edges.

We also consider \emph{orthogonal unit circle intersection graphs},
that is, orthogonal circle intersection graphs with a representation
that consists only of unit circles.  We show that these graphs are a
proper subclass of penny graphs.  It is NP-hard to recognize penny
graphs~\cite{eades1996}.  We modify the NP-hardness proof of Di
Battista et al.~\cite[Section 11.2.3]{gd-book}, which uses the
\emph{logic engine}, to obtain the NP-hardness of recognizing
orthogonal unit circle intersection graphs
(Theorem~\ref{thm:recognition}).

\section{Preliminaries}
\label{sec:preliminaries}

We will use
the following type of M\"obius transformation~\cite{excursions}.
Let $\alpha$ be a circle having center at $\ce{\alpha}$ and radius 
$\ra{\alpha}$. The \emph{inversion} 
with respect to $\alpha$ is a mapping that maps 
any point $P \neq \ce{\alpha}$ to a point $P'$ on the ray $\ce{\alpha}P$ so that
$
|\ce{\alpha}P'|\cdot|\ce{\alpha}P| = \ra{\alpha}^2
$.
Inversion maps each 
circle not passing through $\ce{\alpha}$ to another circle 
and a circle passing through $\ce{\alpha}$ to a line;
see Fig.~\ref{fig:inversion}.
\begin{figure}[tb]
 \begin{subfigure}[t]{0.29\textwidth}
      \centering
      \includegraphics[page=1]{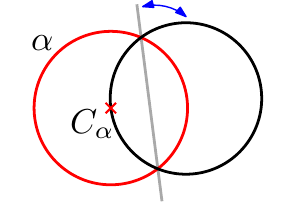}
      \caption{a circle passing through $\ce{\alpha}$ is mapped to a
        line (and vice versa)}
      \label{fig:inversion-1}
 \end{subfigure}
  \hfill
 \begin{subfigure}[t]{0.28\textwidth}
      \centering
      \includegraphics[page=3]{inversion}
      \caption{a circle not passing through $\ce{\alpha}$
      is mapped to another circle}
      \label{fig:inversion-2}
 \end{subfigure}
 \hfill
 \begin{subfigure}[t]{0.35\textwidth}
      \centering
      \includegraphics[page=2]{inversion}
      \caption{constructing the inversion~$P'$ of a point~$P$ w.r.t.\
        $\alpha$ via a circle~$\beta$ orthogonal to~$\alpha$}
      \label{fig:eqasy-construction}
 \end{subfigure}
 \caption{Examples of inversion}
 \label{fig:inversion}
\end{figure}
Inversion and orthogonal circles are closely related.  For example, in
order to construct the image~$P'$ of some point~$P$ that lies inside
the inversion circle~$\alpha$, consider the intersection points~$X$
and~$Y$ of~$\alpha$ and the line that is orthogonal to the line
through $\ce{\alpha}$ and~$P$ in~$P$; see
Fig.~\ref{fig:eqasy-construction}.  The point~$P'$ then is simply the
center of the circle~$\beta$ that is orthogonal to~$\alpha$ and goes
through~$X$ and~$Y$.  This follows from the similarity of the
orthogonal triangles $\triangle \ce{\alpha}XP'$ and $\triangle
\ce{\alpha}XP$.
A useful property of inversion, as of any other M\"obius transformation,
is that it preserves angles. 
Using inversion we can easily show several properties of orthogonal circles.
\begin{lemma}
  \label{lem:no4circles}
  No orthogonal circle intersection graph contains a $K_4$.  In other
  words, in an arrangement of orthogonal circles there cannot be four
  pairwise orthogonal circles.
\end{lemma}
\begin{proof}
  Assume that there are four pairwise orthogonal circles $\alpha$,
  $\beta$, $\gamma$, and $\delta$.  Let $X$ and $Y$ be the
  intersection points of $\alpha$ and~$\beta$.  Consider the inversion
  with respect to a circle~$\sigma$ centered at~$X$.  The images of
  $\alpha$ and $\beta$ are orthogonal lines $\alpha'$ and $\beta'$
  that intersect at~$Y'$, which is the image of~$Y$; see
  Fig.~\ref{fig:triple}.  The image of $\gamma$ is a circle $\gamma'$
  centered at~$Y'$ but so is the image~$\delta'$ of~$\delta$.
  Thus~$\gamma'$ and~$\delta'$ are either disjoint or equal, but not
  orthogonal to each other, a contradiction.
\end{proof}
\begin{figure}[tb]
\minipage[b]{0.67\textwidth}
 \begin{subfigure}[t]{0.47\textwidth}
      \centering
      \includegraphics[page=1]{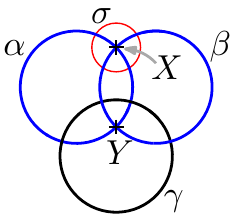}
      \caption{}
      \label{fig:triple-with-invcircle}
 \end{subfigure}
  \hfill
 \begin{subfigure}[t]{0.47\textwidth}
      \centering
      \includegraphics[page=2]{no_concave_side}
       \caption{}
      \label{fig:image-of-triple}
 \end{subfigure}
 \caption{(a) Three pairwise intersecting circles, the red inversion
   circle is centered at $X$; (b) image of the inversion.}
 \label{fig:triple}
\endminipage
\hfill
\minipage[b]{0.28\textwidth}
  \centering
  \includegraphics{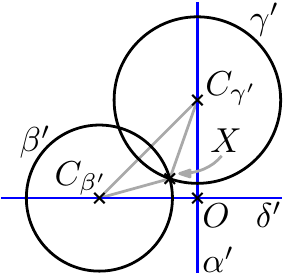}
  \caption{Illustration for the proof of Lemma~\ref{lem:no_induced_c4}}
  \label{fig:no_induced_c4}
\endminipage
\end{figure}

\begin{lemma}
  \label{lem:no_induced_c4}
  No orthogonal circle intersection graph contains an
  induced~$C_4$.  In other words, in an arrangement of orthogonal
  circles there cannot be two pairs of circles such that
  each circle of one pair is orthogonal to each circle of the other
  pair and the circles within the pairs are not orthogonal.
\end{lemma}

\begin{proof}
 Assume there are two pairs $(\alpha, \beta)$ and
 $(\gamma, \delta)$ of circles such that
 the circles within each pair do not intersect each other and 
 each circle of one pair intersects both circles of the other pair. 
 Consider an inversion
 via a circle $\sigma$ centered at one of the intersection points of the circles $\alpha$ and $\delta$.
 In the image they will become lines $\alpha'$ and $\delta'$.
 The image $\beta'$ of the circle $\beta$ must intersect $\delta'$ but
 not $\alpha'$, therefore, its center must lie on the line $\delta'$
 and it should be to one side of the line $\alpha'$; see Fig.~\ref{fig:no_induced_c4}. 
 Similarly the center of the image $\gamma'$ of the circle $\gamma$ 
 must lie on the line $\alpha'$
 and $\gamma'$ should be to one side of the line $\delta'$. 
 Shift the drawing so that the intersection of $\alpha'$ and $\delta'$
 is at the origin $O$ and observe that the triangle 
 $\triangle \ce{\beta'}O\ce{\gamma'}$
 is orthogonal, where $\ce{\beta'}$ and $\ce{\gamma'}$ are 
 the centers of the circles $\beta'$ and $\gamma'$.  Let
 $X$ be the intersection point of these circles that is closer to the
 origin.  This point~$X$ is contained in the triangle $\triangle
 \ce{\beta'}O\ce{\gamma'}$.  Therefore the triangle $\triangle
 \ce{\beta'}X\ce{\gamma'}$ cannot be orthogonal---a
 contradiction.
\end{proof}

A \emph{pencil} is a family of circles 
who share a certain characteristic.
In a \emph{parabolic} pencil all circles have one point in common, 
and thus are all tangent to each other; %
see Fig.~\ref{fig:apollonian-just-parabolic}.   
In an \emph{elliptic} pencil all circles go through two given points;
see the gray circles in Fig.~\ref{fig:apollonian}.  
In a \emph{hyperbolic} pencil all circles are orthogonal 
to a set of circles that go through two given points, that is, to some
elliptic pencil; see the black circles in Fig.~\ref{fig:apollonian}.  

For an elliptic pencil whose circles share two points~$A$~and~$B$ and the
corresponding hyperbolic pencil, 
the circles in the hyperbolic pencil possess several properties useful
for our purposes~\cite{excursions}.
Their centers are collinear and they consist of non-intersecting
circles that form two nested structures of circles,
one containing~$A$, the other one containing $B$ in its interior; see
Fig.~\ref{fig:apollonian}.

Two pencils of circles such that each circle in one pencil is orthogonal to each circle in the other are called
\emph{Apollonian circles}. There can be two such combinations of pencils, that is, one with two parabolic pencils and one with
an elliptic and a hyperbolic pencil. We focus on the latter since such Apollonian circles 
contain  arbitrarily large arrangements of orthogonal circles, that is, two orthogonal circles from
the elliptic pencil and arbitrary many circles from the hyperbolic pencil.
Equivalently, such Apollonian circles are
an inversion image of a family of concentric circles centered at some
point $X$ and concurrent lines passing through $X$; see
Fig.~\ref{fig:apollonian-inversion}.  We use this equivalence in the
next proof.

\begin{figure}[tb]
  \minipage[b]{0.31\textwidth}
      \centering
      \includegraphics{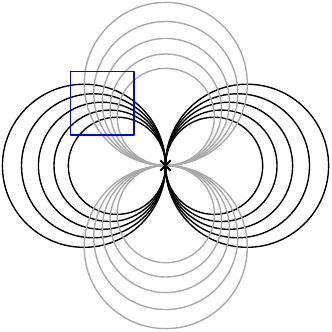}
      \caption{Apollonian circles consisting of two parabolic pencils
        of circles (one in black, the other in gray).}
      \label{fig:apollonian-just-parabolic}
  \endminipage
  \hfill
  \minipage[b]{0.65\textwidth}
    \begin{subfigure}{0.58\textwidth}
      \centering
      \includegraphics{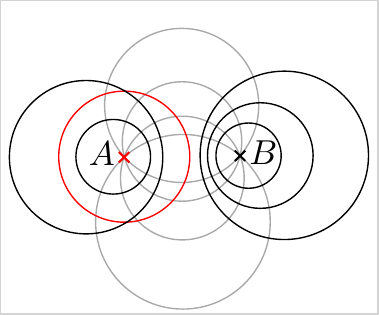}
      \caption{}	
      \label{fig:apollonian}
   \end{subfigure}
   \hfill
   \begin{subfigure}{0.38\textwidth}
      \centering
      \includegraphics{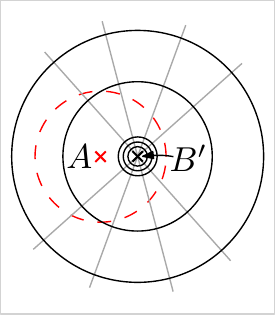}
      \caption{}
      \label{fig:apollonian-inversion}
   \end{subfigure} 
   \caption{(a) Apollonian circles consisting of an elliptic (in gray)
     and hyperbolic (in black) pencil of circles; (b) its inversion
     via a circle centered at $A$ (in red).
     }
  \endminipage
\end{figure}

\begin{lemma}
  \label{lem:apollonian-structure}
  Three circles such that one is orthogonal to the two others belong to the
  same family of Apollonian circles. Two sets of circles such that each circle in
  one set is orthogonal to each circle in the other set and each set has
  at least two circles belong to the same family of Apollonian circles.
  In particular the set belonging to the elliptic pencil can contain at most two 
  circles.
\end{lemma}

\begin{proof}
 Consider three circles such that one is orthogonal to two others.
 If all three are pairwise orthogonal, then their inversion via a circle centered at one
 of their intersection points (see Fig.~\ref{fig:triple-with-invcircle}) 
 is two perpendicular lines and 
 a circle centered at their intersection point (see Fig.~\ref{fig:image-of-triple}),
 therefore, they belong to the same family
 of Apollonian circles.
 If two circles do not intersect, then
 by \cite[Theorem~13]{excursions}, 
 it is always possible to invert them into two concentric circles. Since
 inversion preserves angles, the image of the third circle must be orthogonal 
 to both concentric circles and therefore it must be a straight line passing through the center of
 both circles. Therefore, the three circles belong to the same family
 of Apollonian circles.
 
 Consider now two sets $S_1$ and $S_2$ of circles such that each circle in
 one set is orthogonal to each circle in the other set and each set has
 at least two circles. By Lemma~\ref{lem:no_induced_c4}
 there must be two circles  $\alpha$ and $\beta$ in
 one of the sets, say $S_1$, that are orthogonal. 
 Consider an inversion via a circle $\sigma$ centered
 at one of the intersection points $X$
 of the circles $\alpha$ and $\beta$.
 In the image they will become orthogonal 
 lines $\alpha'$ and $\beta'$ intersecting at a point $Y$.
 Because inversion preserves angles,
 the image of each circle in $S_2$ is a circle centered at $Y$.
 Since $S_2$ contains at least two circles, the image of each circle in 
 $S_1$ must be orthogonal to two circles centered at $Y$, therefore,
 it must be a straight line passing through $Y$.
 Thus, the circles in $S_1$ and $S_2$ belong to the same family
 of Apollonian circles and $S_1$ contains at most two circles.
\end{proof}

Because each triangular or quadrangular face consists of 
either three circles such that one is orthogonal to two others or
two pairs of circles such that each circle in 
one pair is orthogonal to each circle in the other pair, 
we obtain the following observation from Lemma~\ref{lem:apollonian-structure}.
\begin{observation}
  \label{obs:face-shape}
  In any arrangement of orthogonal circles, each triangular 
  and each quadrangular face 
  is formed by Apollonian circles.
\end{observation}

\section{Arrangements of Orthogonal Circles}
\label{sec:complexity}

In this section we study the number of faces of an arrangement 
of orthogonal circles. In Section~\ref{sec:counting-faces},
we give a bound on the total number of faces. In 
Section~\ref{sec:counting-small-faces},
we separately bound the number of faces formed by two and three edges.

Let \arr be an arrangement of orthogonal circles in the plane.  By a
slight abuse of notation, we will say that a circle~$\alpha$
\emph{contains} a geometric object~$o$ and mean that the disk bounded
by~$\alpha$ contains~$o$.  We say that a circle $\alpha \in \arr$ is
\emph{nested} in a circle~$\beta \in \arr$ if $\alpha$ is contained
in~$\beta$.  We say that a circle $\alpha \in \arr$ is nested
\emph{consecutively} in a circle $\beta \in \arr$ if $\alpha$ is
nested in~$\beta$ and there is no other
circle~$\gamma\in\arr$ such that $\alpha$ is nested
in~$\gamma$ and~$\gamma$ is nested in $\beta$.
Consider a subset $S \subseteq \arr$ of maximum cardinality such that
for each pair of circles one is nested in the other.
The innermost circle~$\alpha$ in $S$ 
is called a \emph{deepest} circle in \arr; see
Fig.~\ref{fig:maximum-nested}.

\begin{lemma}
  \label{lem:not-too-many-bigger-circles}
  Let $\alpha$ be a circle of radius~$\ra{\alpha}$, and let $S$ be a set
  of circles orthogonal to~$\alpha$.  If $S$ does not contain nested
  circles and each circle in~$S$ has radius at least~$\ra{\alpha}$, then
  $|S| \le 6$.  Moreover, if $|S| = 6$, then all circles in~$S$ have
  radius~$\ra{\alpha}$ and $\alpha$ is contained in the union of the
  circles in~$S$.
\end{lemma} 
\begin{proof}
  Let $\ce{\alpha}$ be the center of $\alpha$.  Consider any two
  circles $\beta$ and~$\gamma$ in~$S$ with centers $\ce{\beta}$ and
  $\ce{\gamma}$ and with radii~$\ra{\beta}$ and~$\ra{\gamma}$,
  respectively.  Since $\ra{\beta} \ge \ra{\alpha}$ and
  $\ra{\gamma} \ge \ra{\alpha}$, the edge $\ce{\beta}\ce{\gamma}$
  is the longest edge of the triangle
  $\triangle\ce{\beta}\ce{\alpha}\ce{\gamma}$; see
  Fig.~\ref{fig:not-too-many-bigger-circles}.  So the angle
  $\angle\ce{\beta}\ce{\alpha}\ce{\gamma}$ is at least $\pi/3$.
  Thus, $|S| \le 6$.

  Moreover, if $|S|=6$ then, for each pair of circles $\beta$
  and~$\gamma$ in~$S$ that are consecutive in the circular ordering
  of the circle centers around~$\ce{\alpha}$, it holds that
  $\angle\ce{\beta}\ce{\alpha}\ce{\gamma} = \pi/3$.  This is only
  possible if $\ra{\beta}=\ra{\gamma}=\ra{\alpha}$.  Thus, all the
  circles in~$S$ have radius~$\ra{\alpha}$ and %
  $\alpha$ is contained in the union of the circles in~$S$; see
  Fig.~\ref{fig:contradiction-1}.
\end{proof}

\begin{figure}[tb]
  \minipage{0.48\textwidth}
    \centering
    \includegraphics{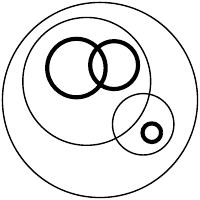}
    \caption{Deepest circles in bold}
    \label{fig:maximum-nested}
  \endminipage
  \hfill
  \minipage{0.48\textwidth}
    \centering
      \includegraphics{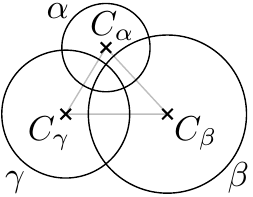}
      \caption{%
        $\angle\ce{\beta}\ce{\alpha}\ce{\gamma} \ge \pi/3$}
       \label{fig:not-too-many-bigger-circles}
  \endminipage
\end{figure}

\subsection{Bounding the Number of Faces}
\label{sec:counting-faces}

\begin{restatable}{theorem}{orthogonalCircleArrangementFaces}
  \label{thm:orthogonalCircleArrangements}
  Every arrangement of $n$ orthogonal circles has at most $14n$ intersection points
  and $15n+2$ faces. 
\end{restatable}
The above theorem (whose formal proof is at the end of the section)
follows from the fact that any arrangement of orthogonal circles contains
a circle~$\alpha$ with at most seven \emph{neighbors} (that is, circles that
are orthogonal to~$\alpha$).
\begin{lemma}
\label{lem:max-deg-circle}
Every arrangement of orthogonal circles 
has a circle that is orthogonal to at most seven other circles.
\end{lemma} 
\begin{proof}
If no circle is nested within any other, Lemma~\ref{lem:not-too-many-bigger-circles} implies that the smallest circle has at most six neighbors, and we are done. 

So, among the deepest circles in~\arr, consider a circle~$\alpha$
with the smallest radius.  Let~$r_\alpha$ be the radius of~$\alpha$.
Note that~$\alpha$ is nested in at least one circle. 
Let $\beta$ be a circle
such that $\alpha$ and $\beta$ are consecutively nested.
Denote the set of all circles in \arr that are orthogonal to $\alpha$
but not to~$\beta$ by~$S_{\alpha}$. 
All circles in $S_{\alpha}$ are nested in $\beta$.
Since $\alpha$ is a deepest circle, $S_{\alpha}$ contains no nested circles;
see Fig.~\ref{fig:max-deg-circle}.
Since the radius of every circle in $S_{\alpha}$ is at least~$r_\alpha$,
Lemma~\ref{lem:not-too-many-bigger-circles} ensures that $S_{\alpha}$
contains at most six circles.
Given the structure of Apollonian circles (Lemma~\ref{lem:apollonian-structure}),
there can be at most two
circles that intersect both~$\alpha$ and~$\beta$.
This together with Lemma~\ref{lem:not-too-many-bigger-circles} immediately implies
that $\alpha$ cannot be orthogonal to more than eight circles.
In the following we show that there can be at most seven such circles.

If there is only one circle intersecting both $\alpha$ and $\beta$,
then $\alpha$ is orthogonal to at most seven circles in total, and we
are done.

Otherwise, there are two circles orthogonal to both~$\alpha$
and~$\beta$.  Let these circles be~$\gamma_1$ and~$\gamma_2$.  We
assume that~$S_{\alpha}$ contains exactly six circles.  Hence, by
Lemma~\ref{lem:not-too-many-bigger-circles}, all circles
in~$S_{\alpha}$ have radius~$r_\alpha$.  Let $S_{\alpha} =
(\delta_0, \ldots, \delta_5)$ be ordered clockwise around $\alpha$
so that every two circles~$\delta_i$ and~$\delta_j$ with $i \equiv
j+1 \bmod 6$ are orthogonal.  

Let~$X$ and~$Y$ be the intersection points of~$\gamma_1$
and~$\gamma_2$; see Fig.~\ref{fig:max-deg-circle}.  Note that, by the
structure of Apollonian circles, one of the intersection points, say
$X$, must be contained inside~$\alpha$, whereas the other
intersection point~$Y$ must lie in the exterior of~$\beta$.  Since
the circles in~$S_{\alpha}$ are contained in~$\beta$, none of them
contains~$Y$.  Further, no circle~$\delta_i$ in~$S_{\alpha}$
contains~$X$, as otherwise the circles $\delta_i$, $\alpha$,
$\gamma_1$, and $\gamma_2$ would be pairwise orthogonal,
contradicting Lemma~\ref{lem:no4circles}.  Recall that, by
Lemma~\ref{lem:not-too-many-bigger-circles}, $\alpha$ is contained in
the union of the circles in~$S_{\alpha}$.  Since $X$ is not contained
in this union, $\gamma_1$ intersects two different circles
$\delta_i$ and $\delta_j$, and $\gamma_2$ intersects two different
circles~$\delta_k$ and~$\delta_l$.  
Note that $\gamma_1$ and $\gamma_2$ cannot
intersect the same circle~$\varepsilon$ in~$S_{\alpha}$,
because~$\varepsilon$, $\alpha$, $\gamma_1$, and $\gamma_2$ would
be pairwise orthogonal, contradicting
Lemma~\ref{lem:no4circles}.  Therefore, the indices $i$, $j$,
$k$, and~$l$ are pairwise different.

We now consider possible values of the indices~$i$, $j$, $k$, and~$l$,
and show that in each case we get a contradiction to
Lemma~\ref{lem:no4circles} or Lemma~\ref{lem:no_induced_c4}.  If $j
\equiv i+1 \bmod 6$, then $\gamma_1$, $\alpha$, $\delta_i$, and
$\delta_j$ would be pairwise orthogonal, contradicting
Lemma~\ref{lem:no4circles}; see Fig.~\ref{fig:contradiction-1}.
If $j \equiv i+2 \bmod 6$, then $\gamma_1$, $\delta_i$,
$\delta_{i+1}$, and~$\delta_j$ would form an induced~$C_4$ in the
intersection graph; see Fig.~\ref{fig:contradiction-2}.  This would
contradict Lemma~\ref{lem:no_induced_c4}.  If $j \equiv i+3
\bmod 6$ and $k \equiv l+3 \bmod 6$, then either $k \equiv i+1 \bmod
6$ or $i \equiv l+1 \bmod 6$; see Fig.~\ref{fig:contradiction-3}.
W.l.o.g., assume the latter and observe that then $\gamma_2$,
$\delta_i$, $\gamma_1$, $\delta_l$ would form an induced~$C_4$, again
contradicting Lemma~\ref{lem:no_induced_c4}.

We conclude that $S_{\alpha}$ contains at most five circles.
Together with~$\gamma_1$ and~$\gamma_2$, at most seven circles are
orthogonal to~$\alpha$.
\end{proof}

\begin{figure}[tb]
  \begin{subfigure}[b]{0.28\textwidth}
    \centering
    \includegraphics{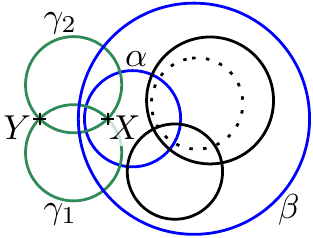}
  \end{subfigure}
  \hfill
  \begin{subfigure}[b]{0.22\textwidth}
    \centering
    \includegraphics[page = 1]{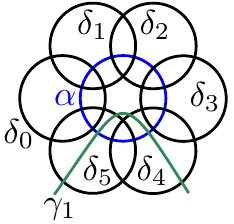}
  \end{subfigure}
  \hfill
  \begin{subfigure}[b]{0.22\textwidth}
    \centering
    \includegraphics[page = 2]{max-deg-circle-contradictions}
  \end{subfigure}
  \hfill
  \begin{subfigure}[b]{0.2\textwidth}
    \centering
    \includegraphics[page = 3]{max-deg-circle-contradictions}
  \end{subfigure}

  \begin{subfigure}[t]{0.28\textwidth}
    \centering
    \caption{the circles of $S_{\alpha}$ are in bold black}
    \label{fig:max-deg-circle}
  \end{subfigure}
  \hfill
  \begin{subfigure}[t]{0.22\textwidth}
    \centering
    \caption{$i=4$, $j=5$}
    \label{fig:contradiction-1}
  \end{subfigure}
  \hfill
  \begin{subfigure}[t]{0.22\textwidth}
    \centering
    \caption{$i=3$, $j=5$}
    \label{fig:contradiction-2}
  \end{subfigure}
  \hfill
  \begin{subfigure}[t]{0.2\textwidth}
    \centering
    \caption{$i=2$, $j=5$, \mbox{$l=1$,} $k=4$}
    \label{fig:contradiction-3}
  \end{subfigure}

  \caption{Illustrations to the proof of Lemma~\ref{lem:max-deg-circle}}
  \label{fig:max-deg-circle-contradictions}
\end{figure}

Using the lemma above and Euler's formula, we now can prove
Theorem~\ref{thm:orthogonalCircleArrangements}.

\begin{proof}[of Theorem~\ref{thm:orthogonalCircleArrangements}]
  Let \arr be an arrangement of orthogonal circles.  By
  Lemma~\ref{lem:max-deg-circle}, \arr contains a circle~$\alpha$
  orthogonal to at most seven circles.  The circle~$\alpha$ yields at
  most $14$ intersection points.  By induction, the whole arrangement
  has at most $14n$ intersection points.
 
  Consider the planarization $G'$ of~\arr, and let $n'$, $m'$, $f'$,
  and~$c'$ denote the numbers of vertices, edges, faces, and connected
  components of~$G'$,
  respectively.  Since every vertex in the planarization corresponds
  to an intersection, the resulting graph is $4$-regular and therefore
  $m' = 2n'$.  By Euler's formula, we obtain $f' = n'+ 1 + c'$.  This
  yields $f' \le 15n + 1$ since $n' \le 14n$ and $c' \le n$.
\end{proof}

\subsection{Bounding the Number of Small Faces}
\label{sec:counting-small-faces}

In the following
we study the number of faces of each type,
that is, the number of 
digonal, triangular, and quadrangular faces.
We begin with some notation.
Let~$\arr$ be an arrangement of orthogonal circles in the plane.
Let $S$ be some subset of the circles of $\arr$. 
A face in $S$ is called a region in $\arr$ formed by $S$; see for
instance Fig.~\ref{fig:region}.
Note that each face of $\arr$ is also a region.

Let $s$ be the region formed by some circular arcs $a_1,a_2,\dots,a_k$
enumerated in counterclockwise order around~$s$.  For an arc $a_i$
with $i \in \{1,\dots,k\}$, let $\alpha$ be the circle that
supports~$a_i$.  If $\ce{\alpha}=(x_{\alpha}, y_{\alpha})$ is the center of
$\alpha$ and $\ra{\alpha}$ its radius, we can write $\alpha$ as
$\big\{\ce{\alpha}+\ra{\alpha}(\cos t,\sin t) \colon
t\in[0, 2\pi]\big\}$.  Let $u$ and $v$ be the endpoints
of~$a_i$ so that we meet~$u$ first when we traverse~$s$
counterclockwise when starting outside of~$a_i$.
Let $u = \ce{\alpha} + \ra{\alpha}(\cos{t_1},\sin{t_1})$ and    
$v = \ce{\alpha} + \ra{\alpha}(\cos{t_2},\sin{t_2})$.
We say that the region $s$ \emph{subtends} an angle in the circle $\alpha$
of size $\subtend{s}{a_i} = t_2 - t_1$ with respect to the arc $a_i$.
Note that $\subtend{s}{a_i}$ is negative if $a_i$ forms a concave side of $s$.
If the circle $\alpha$ forms only one side of the region $s$, then we just
say that the region $s$ \emph{subtends} an angle in the circle $\alpha$
of size $\subtend{s}{\alpha} = t_2 - t_1$.
Moreover, if $s$ is a digonal region, that is, it is formed by only two circles 
$\alpha$ and $\beta$, then
we simply say that  $\beta$ subtends an angle of $\subtend{\beta}{\alpha}  = t_2 - t_1$ in $\alpha$ to mean
$\subtend{s}{\alpha}$.

By  \emph{total angle} we denote the sum of subtended angles by $s$ with respect
to all the arcs that form its sides, that is,  $\sum_{i=1}^k \subtend{s}{a_i}$.

\begin{figure}[tb]
  \minipage[b]{0.37\textwidth}
      \centering
      \includegraphics{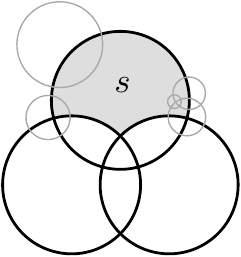}
      \caption{Region $s$ is a face in the arrangement of the bold circles}
      \label{fig:region}
  \endminipage
  \hfill
  \minipage[b]{0.53\textwidth}
      \centering
      \includegraphics{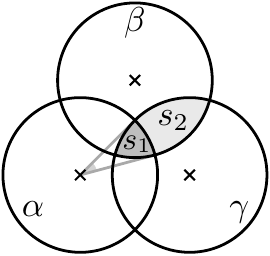}
      \caption{Angles subtended by the regions $s_1$ and $s_2$ in the
        circle $\alpha$; $\subtend{s_1}{\alpha} =
        -\subtend{s_2}{\alpha}$}
      \label{fig:subtention}
  \endminipage
\end{figure}
 
We now give an upper bound on the
number of digonal and triangular faces in an arrangement $\arr$
of $n$ orthogonal circles. 
The tool that we utilize in this section is the Gauss--Bonnet
formula~\cite{gauss-bonnet} which,
in the restricted case of orthogonal circles in the plane, states that,
for every region~$s$ formed by some circular arcs $a_1,a_2,\dots,a_k$,
it holds that
\begin{equation*}
 \sum_{i=1}^k \subtend{s}{a_i} + \frac{\D k\pi}{\D 2} = 2\pi.
\end{equation*}
This formula implies that each digonal or triangular face subtends a
total angle of size $\pi$ and of size $\pi/2$, respectively.  Thus, we
obtain the following bounds.

\begin{theorem}
\label{thm:small-faces}
 Every arrangement of $n$ orthogonal circles has at most $2n$ digonal~faces and
 at most $4n$ triangular faces.
\end{theorem}
\begin{proof}
Because faces do not overlap, each digonal or triangular face uses a unique convex arc 
of a circle bounding this face. 
Therefore, the sum of angles subtended by digonal or triangular faces formed by the same circle
must be at most $2\pi$.
Analogously, the sum of total angles over all digonal or triangular faces cannot exceed $2n\pi$.
By the Gauss--Bonnet formula each digonal or triangular face subtends
a total angle of size $\pi$ or $\pi/2$, respectively.
This gives an upper bound of $2n$ on the number of digonal faces and
an upper bound of $4n$ on the number of triangular faces.
\end{proof}

Theorem~\ref{thm:small-faces} 
can be generalized to all convex orthogonal closed curves since the
Gauss--Bonnet formula does not require curves to be circular. 
In contrast to this, for example, a grid made of axis parallel 
rectangles has quadratically many quadrangular faces. 
This makes circles a special subclass of convex orthogonal 
closed curves. 
We refer to the full version for more details \cite{full-version-arxiv}.

The Gauss--Bonnet formula does not help us to get an upper bound
on the number of quadrangular faces. 
Using Observation~\ref{obs:face-shape}, however, it is possible to
restrict the types of quadrangular faces to several shapes and obtain
bounds on the number of faces of each type.  Apart from being
interesting in its own right, such a bound also provides a bound on
the total number of faces in an arrangement of orthogonal circles.
Namely, since the average degree of a face in an arrangement of
orthogonal circles is~4, a bound on the number of faces of degree at
most~4 gives a bound on the number of all faces in the arrangement
(via Euler's formula).  Unfortunately, the bound on the number of
quadrangular faces that we achieved
was $17n$ and thus higher than the bound $15n+2$ that we now have for
the number of \emph{all} faces in an arrangement of $n$ orthogonal
circles.

\section{Intersection Graphs of Orthogonal Circles}
\label{sec:ocig}

Given an arrangement~$\arr$ of orthogonal circles, consider its
\emph{intersection} graph, which is the graph with vertex set~\arr
that has an edge between any pair of intersecting circles in \arr.
Lemmas~\ref{lem:no4circles} and~\ref{lem:no_induced_c4} imply
that such a graph does not contain any~$K_4$ and any induced~$C_4$.
We show that such graphs can be non-planar (Lemma~\ref{lem:non-planar}),
then we bound their edge density (Theorem~\ref{thm:linear}),
and finally we consider the intersection graphs arising from
orthogonal \emph{unit} circles (Theorem~\ref{thm:recognition}).

\begin{lemma}
\label{lem:non-planar}
  For every $n$, there is an intersection graph of orthogonal circles
  that contains~$K_n$ as a minor.  The representation uses circles of
  three different radii.
\end{lemma}
\begin{proof}
  Let a \emph{chain} be an arrangement of orthogonal circles whose
  intersection graph is a path.  We say that two chains~$C_1$
  and~$C_2$ \emph{cross} if two disjoint circles~$\alpha$
  and~$\beta$ of one chain, say~$C_1$, are orthogonal to the same
  circle~$\gamma$ of the other chain~$C_2$; see
  Fig.~\ref{fig:chain-crossing} (left).  If two chains cross, their
  paths in the intersection graph are connected by two edges; see the
  dashed edges in Fig.~\ref{fig:chain-crossing} (right).
 
  Consider an arrangement of $n$ rectilinear paths embedded on a grid
  where each pair of curves intersect exactly once; see the inset in
  Fig.~\ref{fig:k5-paths}.  We convert the arrangement of paths into
  an arrangement of chains such that each pair of chains crosses; see
  Fig.~\ref{fig:k5-paths}.  Now consider the intersection graph of the
  orthogonal circles in the arrangement of chains.  If we contract
  each path in the intersection graph that corresponds to a chain, we
  obtain~$K_n$.
\end{proof}

\begin{figure}[tb]
\centering
  \begin{subfigure}[b]{0.55\textwidth}
    \centering
    \includegraphics[page = 1]{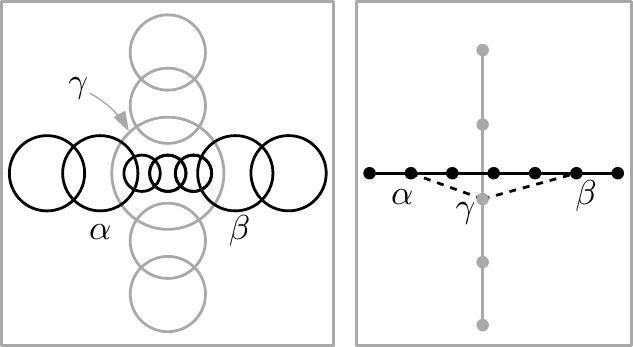}
    \caption{a chain crossing and its intersection graph}
    \label{fig:chain-crossing}
  \end{subfigure}
  \hfill
  \begin{subfigure}[b]{0.43\textwidth}
    \centering
    \includegraphics[page=2,scale=.9]{proper-chain-intersection}
    \caption{pairwise intersecting paths (see inset) and the
      corresponding chains in an orthogonal circle representation}
    \label{fig:k5-paths}
  \end{subfigure}

  \caption{Construction of an orthogonal circle intersection graph
    that contains $K_n$ as a minor (here $n=5$).}
\end{figure}
 
Next, we discuss the density of orthogonal circle intersection graphs.
Gy\'arf\'as et al.~\cite{ghs-lcc4fg-Combin02} have shown that any
$C_4$-free graph on $n$ vertices with average degree at least~$a$ has
clique number at least $a^2/(10n)$.  Due to
Lemma~\ref{lem:no4circles}, we know that orthogonal circle
intersection graphs have clique number at most~3.  Thus, their average
degree is bounded from above by~$\sqrt{30n}$, leading to at most $\sqrt{7.5}n^{\frac{3}{2}}$ edges in total.   
However, Lemma~\ref{lem:max-deg-circle} implies the following stronger bound.

\begin{theorem}
  \label{thm:linear}
  The intersection graph of a set of $n$ orthogonal circles has at
  most $7n$ edges.
\end{theorem}

\begin{proof}
  The geometric representation of an orthogonal circle intersection
  graph is an arrangement of orthogonal circles.  By
  Lemma~\ref{lem:max-deg-circle}, an arrangement of
  $n$ orthogonal circles always has a circle orthogonal to at most seven circles.
  Therefore, the corresponding intersection graph always
  has a vertex of degree at most seven.
  Thus, it has at most $7n$ edges.
\end{proof}

The remainder of this section concerns a natural subclass of
orthogonal circle intersection graphs, the orthogonal \emph{unit} 
circle intersection graphs.  Recall that these are orthogonal circle
intersection graphs with a representation that consists of unit
circles only.  As Fig.~\ref{fig:shrinking} shows, every representation
of an orthogonal unit circle intersection graph can be transformed (by
scaling each circle by a factor of $\sqrt{2}/2$) into a representation
of a \emph{penny graph}, that is, a contact graph of equal-size disks.
Hence, every orthogonal unit circle intersection graph is a penny
graph~-- whereas the converse is not true.  For example, $C_4$ or the
5-star are penny graphs but not orthogonal unit circle intersection
graphs (see Fig.~\ref{fig:penny-c4}).

Orthogonal unit circle intersection graphs being penny graphs implies
that they inherit the properties of penny graphs, e.g., their maximum
degree is at most six and their edge density is at most
$\lfloor 3n - \sqrt{12n-6} \rfloor$, where $n$ is the number of
vertices~\cite[Theorem 13.12, p. 211]{book-pach}.  
Because triangular grids are orthogonal unit circle intersection 
graphs, this upper bound is tight.  

\begin{figure}[tb]
  \centering
  
  \begin{subfigure}[t]{0.42\textwidth}
    \centering
    \includegraphics[page=1]{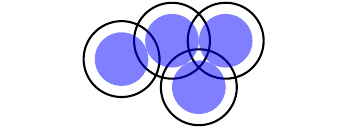}
    \caption{all orthogonal unit circle intersection graphs are penny
      graphs}
    \label{fig:shrinking}
  \end{subfigure}
  \hfil
  \begin{subfigure}[t]{0.42\textwidth}
    \centering
    \includegraphics[page=2]{shrinking}
    \caption{penny graphs that aren't orthogonal unit circle
      intersection graphs}
    \label{fig:penny-c4}
  \end{subfigure}
     
  \caption{Penny graphs vs. orthogonal unit circle intersection
    graphs}
\end{figure}

As it turns out, orthogonal unit circle intersection graphs share
another feature with penny graphs: their recognition is NP-hard.  The
hardness of penny-graph recognition can be shown using the \emph{logic
  engine}~\cite[Section 11.2]{gd-book}, which simulates an instance of
the Not-All-Equal-3-Sat (\textsc{NAE3SAT}) problem.  
We establish a similar reduction for the recognition of
orthogonal unit circle intersection graphs; the details are in
\lncsarxiv{the full version \cite{full-version-arxiv}}%
{the appendix}.

\begin{restatable}{theorem}{thmrecognition}
  \label{thm:recognition}
  It is NP-hard to recognize orthogonal unit circle intersection
  graphs.
\end{restatable}

\section{Discussions and Open Problems}

In Section~\ref{sec:complexity} we have provided upper bounds for the
number of faces of an orthogonal circle arrangement. 
As for lower bounds on the number of faces, we found only very simple
arrangements containing
$1.5n$ digonal, $2n$ triangular, and $4(n-3)$ quadrangular faces; 
see Figs.~\ref{fig:lower-bound-digons}, \ref{fig:lower-bound-triangles}, 
and~\ref{fig:lower-bound-quadrangles}, respectively.
Can we construct better lower bound examples or improve the upper bounds?

\begin{figure}[tb]
  \begin{subfigure}[b]{0.26\textwidth}
      \centering
      \includegraphics{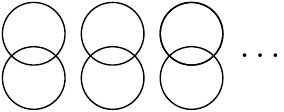}
      \caption{$1.5n$ digonal faces}
      \label{fig:lower-bound-digons}
  \end{subfigure}
  \hfill
  \begin{subfigure}[b]{0.26\textwidth}
      \centering
      \includegraphics{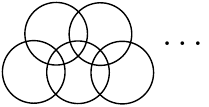}
      \caption{$2n$ triangular faces}
      \label{fig:lower-bound-triangles}
  \end{subfigure}
  \hfill
  \begin{subfigure}[b]{0.37\textwidth}
      \centering
      \includegraphics{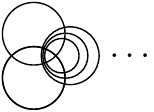}
      \caption{$4(n-3)$ quadrangular faces}
      \label{fig:lower-bound-quadrangles}
  \end{subfigure}
  \caption{Arrangements of $n$ orthogonal circles with many digonal,
    triangular, and quadrangular faces.}
\end{figure}

Recognizing (unit) disk intersection graphs is
$\exists\mathbb{R}$-complete~\cite{kang-ross-sphere-representations}.
But what is the complexity of recognizing (general) orthogonal circle
intersection graphs?

\paragraph{Acknowledgments.} 
We thank Alon Efrat for useful discussions and an anonymous reviewer
for pointing us to the Gauss-Bonnet formula.

\bibliographystyle{splncs04} %
\bibliography{abbrv,refs}

\lncsarxiv{\end{document}}{}

\newpage
\appendix

\renewcommand{\floatpagefraction}{.9}%
\renewcommand{\bottomfraction}{.9}%
\renewcommand{\topfraction}{.9}%

\section*{Appendix: \\
Recognizing Orthogonal Unit Circle Intersection Graphs}

In this section, we show how to realize the \emph{logic engine} with
orthogonal unit circle intersection graphs. The logic engine simulates
the Not-All-Equal-3-Sat (\textsc{NAE3SAT}) problem where a set $C$ of
clauses each containing three literals from a set of boolean variables
$U$ is given and the question is to find a truth assignment to the
variables so that each clause contains at least one true literal and
at least one false literal.

\begin{figure}[hb]
    \centering
    \includegraphics[page=1]{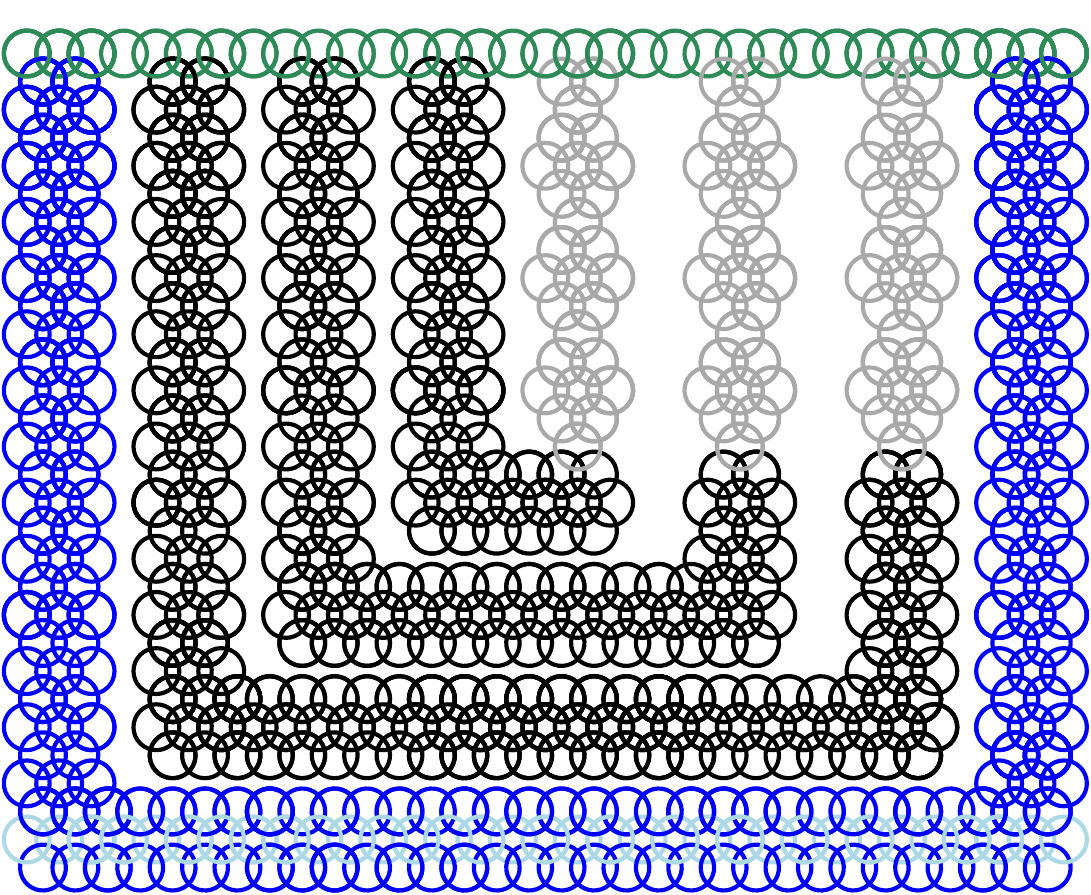}
    \caption{Orthogonal unit circle representation of the universal part of the logic engine; only half of the drawing is present, the other half is symmetric}
    \label{fig:universal}
\end{figure}

\thmrecognition*

\begin{proof}
  We closely follow the description from~\cite[Section 11.2]{gd-book}
  and use their notations and definitions.  The logic engine consists
  of the following parts (we will mostly refer to
  Figs.~\ref{fig:universal} and~\ref{fig:customized} to explain how
  the parts of the logic engine are connected).  The \emph{frame} and
  \emph{armatures} (drawn blue and black respectively in
  Fig.~\ref{fig:universal}, only half of the drawing is illustrated,
  the other half is symmetric with respect to the shaft of the logic engine,
  which is defined below)
  for the logic graph are built of hexagonal
  blocks, as shown in Fig.~\ref{fig:hexagonal-blocks-graph} whose
  orthogonal unit circle intersection representation is shown in
  Fig.~\ref{fig:hexagonal-blocks}. It is easy to see that they are
  uniquely drawable (up to rotation, reflection, and translation)
  since $K_3$ has a unique orthogonal unit circle intersection
  representation.  Each armature corresponds to a variable in $U$.

  A \emph{chain} graph (represented by gray circles in
  Fig.~\ref{fig:universal}) is a sequence of \emph{links}, as shown in
  Fig.~\ref{fig:chain-graph} whose orthogonal unit circle intersection
  representation is shown in Fig.~\ref{fig:chain}.  The number of
  links in a chain corresponds to the number of clauses in $C$.  The
  \emph{shaft} (green in Fig.~\ref{fig:universal}) is a simple path
  and serves as an axle for the armatures, that is, the armatures can
  be flipped around the shaft.  Each armature corresponding to a
  variable $x_j$ has two chains $a_j$ and $\bar a_j$ each suspended
  between one of the ends of the armature and the shaft.  For that
  reason in an orthogonal unit circle intersection representation each
  chain is taut.

  \begin{figure}[tb]
    \begin{subfigure}[t]{0.17\textwidth}
      \centering
      \includegraphics[page=2]{hexagonal-blocks}
      \caption{hexagonal block}
      \label{fig:hexagonal-blocks-graph}
    \end{subfigure}
    \hfill
    \begin{subfigure}[t]{0.27\textwidth}
      \centering
      \includegraphics[page=1]{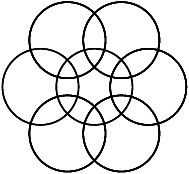}
      \caption{its representation}
      \label{fig:hexagonal-blocks}
    \end{subfigure}
    \hfill
    \begin{subfigure}[t]{0.19\textwidth}
      \centering
      \includegraphics[page=2]{flagged-link}
      \caption{flagged link graph}
      \label{fig:flagged-link-graph}
    \end{subfigure}
    \hfill
    \begin{subfigure}[t]{0.25\textwidth}
      \centering
      \includegraphics[page=1]{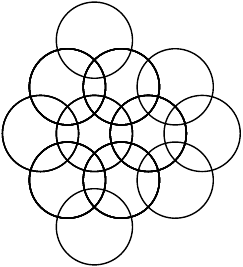}
      \caption{its representation}
      \label{fig:flagged-link}
    \end{subfigure}

    \begin{subfigure}[t]{0.45\textwidth}
      \hspace*{-3ex}
      \includegraphics[page=2,angle=90]{chain}
      \caption{chain graph}
      \label{fig:chain-graph}
    \end{subfigure}
    \hfill
    \begin{subfigure}[t]{0.53\textwidth}
      \centering
      \includegraphics[page=1,angle=90]{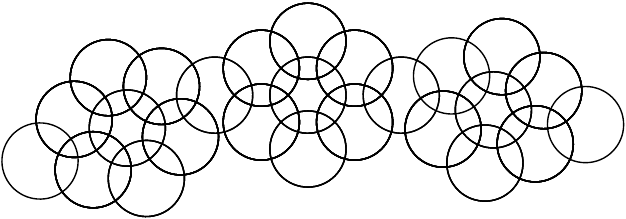}
      \caption{its representation}
      \label{fig:chain}
    \end{subfigure}

    \caption{Gadgets for the logic engine}
  \end{figure}

  So far we have described the \emph{universal} part of the logic
  engine, that is, the part that only depends on the number of clauses
  in $C$ and the number of variables in $U$; it is illustrated in
  Fig.~\ref{fig:universal}.  The frame, armatures, and chain graphs
  have a unique orthogonal unit circle intersection representation up
  to flipping armatures (see Fig.~\ref{fig:universal}), since they are
  built up of hexagonal blocks which are uniquely drawable.  We still
  need to show that the shaft is taut. This is enforced by the bottom
  part of the frame.  Consider the middle horizontal sequence of
  circles in the bottom part of the frame that spans the frame from
  the left side to the right; in light blue in
  Fig.~\ref{fig:universal}.  It is easy to see that the shaft must be
  drawn as this sequence, because it consists of the same number of
  circles and must also span the frame from the left side to the
  right.  Since the sequence is taut, the shaft is also taut. Notice
  that there is still the freedom of flipping each armature together
  with its chains around the shaft, that is, it can take two possible
  positions where one part of the armature is either above or below
  the shaft.  This is the flexibility that allows our logic engine to
  encode a solution of a \textsc{NAE3SAT} instance.

\begin{figure}[tb]
  \centering
  \includegraphics[page=2]{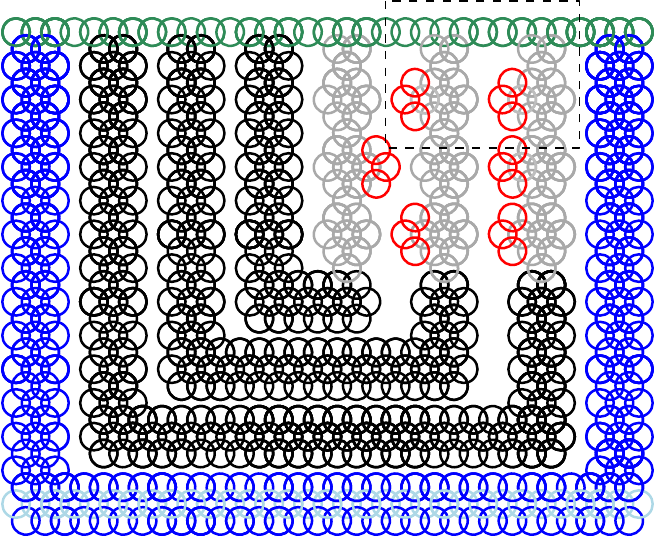}
  \caption{Orthogonal unit circle representation of a customized logic
    engine; only half of the drawing is present. The neighboring flagged links demarcated
    by the dashed rectangle collide if and only if they are flipped so that
    they point towards each other; see Fig.~\ref{fig:collision}.}
  \label{fig:customized}
\end{figure}

  Now let us show how to customize the logic engine according to an
  instance of \textsc{NAE3SAT}.  A chain link graph can be extended to
  a \emph{flagged link} by the addition of three new vertices as shown
  in Fig.~\ref{fig:flagged-link-graph} whose orthogonal unit circle
  representation is shown in Fig.~\ref{fig:flagged-link}. Note that it
  also has a unique drawing.  To simulate the given \textsc{NAE3SAT}
  instance we replace link graphs with flagged link graphs according
  to the incidence between literals and clauses.  If the literal
  $x_j \in U$ appears in clause $c_i \in C$, then link $i$ of chain
  $a_j$ is unflagged.  If the literal $\bar x_j \in U$ appears in
  clause $c_i \in C$, then link $i$ of of chain $\bar a_j$ is
  unflagged.  For an example see Fig.~\ref{fig:customized}.

  It is easy to see that by adjusting the sizes of the frame and the
  armatures we can ensure that in an orthogonal unit circle
  intersection representation of the logic engine two flagged links
  which lie in the same row and are attached to chains of adjacent
  armatures \emph{collide} if and only if they are flipped so that
  they point towards each other; see Fig.~\ref{fig:collision}.
  Similarly we can ensure
  that any flag attached to the chain of the outermost armature
  collides with the frame if it points toward the front edge of the
  frame, and any flag attached to the chain of the innermost armature
  collides with that armature if it points toward the rear.
  Therefore, we can use~\cite[Theorem~11.2]{gd-book} to show that
  the corresponding customized logic engine has an orthogonal unit
  circle representation if and only if the corresponding instance of
  \textsc{NAE3SAT} is a yes-instance.
\end{proof}

\begin{figure}[tb]
  \centering
  \includegraphics{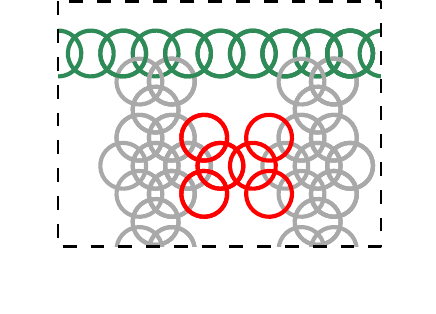}
  \caption{The neighboring flagged links
    collide if and only if they are flipped so that
    they point towards each other.}
  \label{fig:collision}
\end{figure}

\end{document}